\documentclass[a4paper, 11pt]{article}
\usepackage[T1]{fontenc}
\usepackage[utf8]{inputenc}
\usepackage{xspace}
\usepackage{graphicx}
\usepackage{amsmath}
\usepackage{amsthm}
\usepackage{amssymb,mathscinet}
\usepackage{authblk}
\usepackage{enumitem}
\usepackage{thmtools} 
\usepackage{thm-restate}
\usepackage{xcolor}
\usepackage{url}
\usepackage{fullpage}
\usepackage[hidelinks]{hyperref}
\usepackage{stmaryrd}

\newcommand{\eps}{\ensuremath{\varepsilon}\xspace}

\newcommand{\etr}{\ensuremath{\textrm{ETR}}\xspace}
\newcommand{\etrconjunction}{\ensuremath{\textrm{ETR-CONJ}}\xspace}
\newcommand{\etrsmall}{\ensuremath{\textrm{ETR-SMALL}}\xspace}
\newcommand{\etrshift}{\ensuremath{\textrm{ETR-SHIFT}}\xspace}
\newcommand{\etrsquare}{\ensuremath{\textrm{ETR-SQUARE}}\xspace}
\newcommand{\etrinv}{\ensuremath{\textrm{ETR-INV}}\xspace}

\newcommand{\etrami}{\ensuremath{\textrm{ETR-AMI}}\xspace}
\newcommand{\etrcompact}{\ensuremath{\textrm{ETR-COMPACT}}\xspace}

\newcommand{\linearExtension}{linear extension\xspace}
\newcommand{\LLinearExtension}{Linear extension\xspace}
\newcommand{\R}{\ensuremath{\mathbb R}\xspace}
\newcommand{\Q}{\ensuremath{\mathbb Q}\xspace}

\newcommand{\Z}{\ensuremath{\mathbb Z}\xspace}
\newcommand{\N}{\ensuremath{\mathbb N}\xspace}
\newcommand{\ER}{\ensuremath{\exists \mathbb{R}}\xspace}
\newcommand{\ETR}{\ensuremath{\textrm{ETR}}\xspace}
\newcommand{\NP}{\ensuremath{\textrm{NP}}\xspace}
\newcommand{\PSPACE}{\ensuremath{\textrm{PSPACE}}\xspace}
\newcommand{\leqlin}{\leq_{\text{lin}}}
\newcommand{\const}[1]{\ensuremath{\left\llbracket \, #1 \,\right\rrbracket}\xspace}

\newcommand{\mydef}{:=}

\newcommand{\etal}{et al.}

\newtheorem{theorem}{Theorem}
\newtheorem*{theorem*}{Theorem}
\newtheorem{corollary}[theorem]{Corollary}
\newtheorem{example}[theorem]{Example}

\newtheorem{lemmaAlph}{Lemma}

\theoremstyle{definition}

\newtheorem{remark}[theorem]{Remark}
\newtheorem{definition}[theorem]{Definition}

\begin{document}

\title{Dynamic Toolbox for ETRINV}
\author[1]{Mikkel Abrahamsen\footnote{The first author is part of Basic Algorithms Research Copenhagen (BARC), supported by the VILLUM Foundation grant 16582.}}
\author[2]{Tillmann Miltzow\footnote{The second author acknowledges generous support by the NWO Veni grant EAGER.}}
\affil[1]{University of Copenhagen, Denmark.
\texttt{miab@di.ku.dk}}
\affil[2]{ Utrecht University, Netherlands. \texttt{t.miltzow@gmail.com}}

\maketitle

\begin{abstract}
 Recently, various natural algorithmic problems have been shown
 to be \ER-complete. 
The reduction relied in many 
cases on the 
 \ER-completeness of the problem
 \etrinv, which served as a useful 
 intermediate problem.
 Often some strengthening and modification of \etrinv was required.
 This lead to a cluttered situation where no paper included all the previous details. 
 Here, we give a streamlined exposition 
 in a self-contained manner. We also explain and prove various
 universality results regarding \etrinv. 
 
 These notes should not be seen as a research paper with new results.
 However, we do describe some refinements of earlier 
 results which might be useful for future research.
 We plan to extend and update this exposition as seems fit.
\end{abstract}

\section{Introduction}

\subsection{The complexity class \ER}
The \emph{first order theory of the reals} is a set of all true
sentences involving real variables, universal and existential quantifiers, boolean and arithmetic operators, constants $0$ and $1$, parenthesis, equalities and inequalities, i.e., the alphabet is the set
$$\left\{x_1,x_2,\ldots, \forall, \exists, \land,\lor,\lnot, 0 ,1 ,+ ,- ,\cdot,\allowbreak\ (\ ,\ )\ ,=,<,\leq\right\}.$$
A formula is called a \emph{sentence} if it has no free variables, i.e., each variable present in the formula is bound by a quantifier.
Note that using such formulas, we can easily express integer constants (using binary expansion) and powers.
Each formula can be converted to a \emph{prenex form}, which means that it starts with all the quantifiers and is followed by a quantifier-free formula.
Such a transformation changes the length of the formula by at most a constant factor.

The \emph{existential theory of the reals} is the set of all true sentences of the first-order theory of the reals in prenex form with existential quantifiers only, i.e., the sentences are of the form 
$$(\exists x_1 \exists x_2 \ldots \exists x_n)\enspace \Phi(x_1,x_2,\ldots,x_n),$$ 
where $\Phi\mydef \Phi(x_1,x_2,\ldots,x_n)$ is a quantifier-free formula of the first-order theory of the reals with variables $x_1,\ldots,x_n$.
The problem \ETR is the problem of deciding whether a given sentence of the above form is true, and we say that $\Phi$ is an \emph{\ETR formula}.
We define \[V(\Phi) \mydef \{ \mathbf x\in \R^n : \Phi(\mathbf x)  \}.\]
Thus \ETR is the problem of deciding if $V(\Phi)$ is non-empty.
The complexity class \ER consists of  all problems that are reducible to \ETR in polynomial time. It is currently known that \[ \NP \subseteq \ER \subseteq \PSPACE.\]

It is not hard see that the problem \ETR is \NP-hard, 
yielding the first inclusion.
The containment \ER $\subseteq$ \PSPACE is highly non-trivial, 
and it has first been established by Canny~\cite{canny1988some}.
In order to compare the complexity classes $\NP$ and $\ER$, 
we suggest the reader to consider the following two problems.
The problem of deciding whether a given polynomial equation
$Q(x_1,\ldots,x_n)=0$ with integer coefficients has a 
solution with all variables restricted to $\{0,1\}$ is 
easily seen to be \NP-complete.
On the other hand, if the variables are merely restricted to $\R$, 
the problem is \ER-complete~\cite[Proposition 3.2]{matousek2014intersection}.

The \emph{description complexity} or simply \emph{complexity} of an \ETR formula~$\Phi$ is the number of symbols in $\Phi$, and is also denoted $|\Phi|$.
The \emph{complexity} of a semi-algebraic set $S$ is the minimum complexity of a formula $\Phi$ such that $V(\Phi)=S$.

\subsection{Contribution}
 Abrahamsen,  Adamaszek and Miltzow~\cite{ARTETR} introduced the 
 algorithmic problem \etrinv and showed that it is 
 \ER-complete. This was one of the important
 conceptual steps to show that the Art Gallery Problem
 is \ER-complete, for two reasons.
 First, all variables are conveniently bounded
 to the interval $[1/2,2]$. Second,
 it is only necessary to encode inversion constraints
 ($x\cdot y = 1$)
 instead of the more general multiplication constraints
 ($x\cdot y = z$).
 See Section~\ref{sec:INV} for a precise definition.
 
 \begin{figure}[btp]
 \centering
 \includegraphics{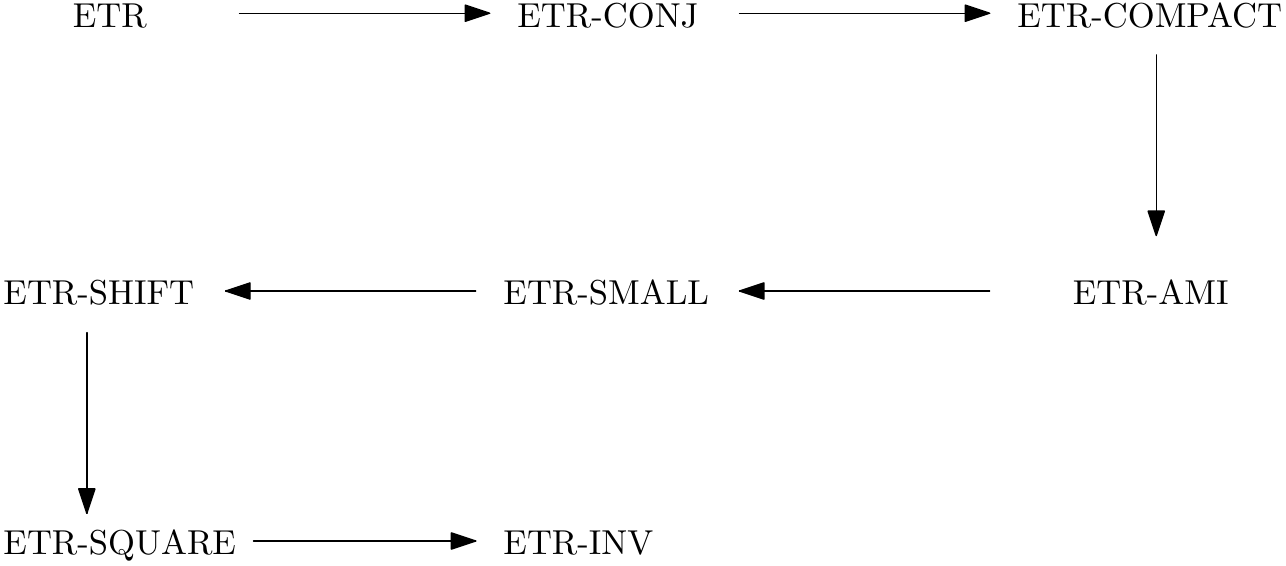}
 \caption{Overview of the problems described in these notes and how they are reduced to each other.
 Almost all reductions preserve rational equivalence and are linear.}
 \label{fig:ReducitonPath}
\end{figure}
 
 In this exposition, we repeat the reduction
 from~\cite{ARTETR}.
 Some refinements are added that 
 might turn out useful in future 
 research~\cite{AreasKleist, NestedPolytopesER, AnnaPreparation}..
 Let us point out the advantages of this exposition.
 
 \paragraph{Self-Containment and exhaustiveness.}
 We explain all details in one coherent exposition. In particular those
 that are scattered around and pointed out in various papers.
 In particular, we also explain some parts that
 were folklore, and only cited in~\cite{ARTETR}.
 (However, we still make use of some facts from real algebraic geometry where we refer to other sources for a proof.)
 
 \paragraph{Universality-type statements.}
 Rather than just \ER-hardness, we 
 also want to make so-called \emph{universality}-type statements,
 which we will explain in Section~\ref{sec:Universality}.
 For that, we have to point out some details
 that were not given in~\cite{ARTETR}.
  
 \paragraph{Running time analysis.}
 We point out the running time of the reductions
 precisely. This can be helpful for so called 
 fine-grained complexity and lower bounds based
 on the Exponential Time Hypothesis.

\paragraph{\LLinearExtension.}
In our reductions, we transform one formula $\Phi$ to another $\Psi$.
In generating $\Psi$, we replace old
variables by new \emph{replacement} variables and 
we add new \emph{auxiliary} variables.
Here, all the replacement variables are
just shifted or scaled versions of the old variables,
and the auxiliary variables are likewise determined by the old
variables.
We capture this by defining the notion of
\linearExtension{}s in Section~\ref{sec:Universality}.
This notion is useful for the so-called 
Picasso-Theorem~\cite{ARTETR} and Kempe's universality theorem.

 \paragraph{Dynamics.}
 It has happened repeatedly that we needed some slightly stronger statement than was provided in the previous paper~\cite{ARTETR}.
 We expect this to happen in the future as well.
 This exposition is supposed to be dynamic and updated, so as to give a complete overview of the known techniques for making \ER-hardness proofs by reductions from \etrinv.

 \paragraph{Modularity.}
 Every reduction is stated as a separate lemma. 
 We hope that in this way it is more convenient to 
 reuse parts of the reductions in forthcoming papers.
 
  \paragraph{Tiny range promise.}
 Interestingly, all variables can 
 be restricted to an arbitrarily small
 range within $[1/2,2]$. This is an important observation
 for a forthcoming paper.

\paragraph{Reductions by diagrams.}
 There were a large number
 of new variables and constraints introduced in the previous reduction to \etrinv.
 It was straight-forward but tedious to
 check that those constraints work correctly and that each new variable was in the correct range.  
We introduce a new type of diagrams to express such reductions involving numerous constraints, which makes it much easier to read and check the reductions.
%
 
 \subsection{Main Result}
 To illustrate our findings, we point out one 
 important theorem. Similar theorems can be derived
 from the reductions presented in this exposition.
 See Section~\ref{sec:Universality} for the definition
 of rational equivalence and \linearExtension.
 \begin{theorem}
 \label{thm:main}
  \etrinv is \ER-complete.
  Furthermore, for every instance $\Phi$ of \etr where $V(\Phi)$ is compact, 
  there is an instance $\Psi$ of \etrinv such that
  $V(\Phi)$ and $V(\Psi)$ are rationally equivalent and $V(\Psi)$ is a \linearExtension of $V(\Phi)$.
 \end{theorem}
 We give two corollaries here to the main theorem.
 \begin{corollary}[Algebraic consequence]
  Let $\alpha$ be an algebraic number. Then there exists
  an instance $\Psi$ of \etrinv, such that 
  $\Psi$ has a solution when the variables are restricted to $\Q[\alpha]$,
  but no solution when the variables are restricted to a field that does not
  contain $\alpha$.
 \end{corollary}
 \begin{proof}
  Let $p \in\Z[x]$ be a univariate polynomial
  with $p(\alpha) = 0$, $p\neq 0$. Furthermore,
  let $\alpha\in [a,b]$ be an interval such that
  $\alpha$ is the only root of $p$ in $[a,b]$.
  Then $p(x) = 0$ and $a\leq x\leq b$, describe
  a compact semi-algebraic set $V$.
  By Theorem~\ref{thm:main} there is an \etrinv
  instance $\Psi$ such that $V$ and $V(\Psi)$ 
  are rationally equivalent.
 \end{proof}
 
 
 \begin{corollary}[Torus]
  There exists
  an instance $\Psi$ of \etrinv such that
  $V(\Psi)$ is homeomorphic to a torus.
 \end{corollary}
 

 \begin{proof}
  The equation 
  \[
  (x^2 + y^2 + z^2 + 5^2 - 1^2)^2 = 4 \cdot 5^2(x^2 + y^2)\]
  describes a torus. 
  This is a compact semi-algebraic set.
  By Theorem~\ref{thm:main} 
  exists an homeomorphic \etrinv
  formula. 
 \end{proof}

\subsection{Related Work}
\label{sec:RealtedWork}
We want to point out that many early parts of 
this reduction can be 
considered folklore. 
Already in 1876 
Kempe~\cite{Kempe1876} exposed ideas
which are the core of our
reduction in Section~\ref{sec:ami}. 
Other important work was done by Shor~\cite{shor1991stretchability} with his
contribution to \ER-hardness of stretchability.
See also~\cite{mnev1988universality, blum2012complexity, DBLP:journals/mst/SchaeferS17,Schaefer2010}.
 
\subsection{Universality-type theorems}
\label{sec:Universality}
A universality statement says that for every
object of type~A exists an object of type~B
such that A and B are equivalent in some sense~C.
In Theorem~\ref{thm:main}, we had compact 
semi-algebraic sets as objects of type~A,
\etrinv formulas are objects of type~B,
and we preserve algorithmic, topological and
algebraic properties. 
Given our reductions, it is possible to make
such statements by replacing~A,~B and~C by 
something else. We think that rational equivalence
and compact semi-algebraic sets are good choices
for two reasons.
First, compact semi-algebraic sets are very versatile.
Of course, it would be nicer to have general, 
rather than compact, semi-algebraic sets, but 
\etrinv cannot encode any open set, so this is not conceivable.
The second reason is that rational equivalence
seems to preserve topological and algebraic properties
in a very strong sense and we are not aware 
how to improve this further.
In this paper, whenever, we refer to polynomials, we
implicitly assume the polynomials to be multivariate
with integer coefficients.

 \begin{definition}[Rational Equivalence]
 Consider two sets $V\subseteq \R^n$ and $W\subseteq \R^m$ and a function $f:V\rightarrow W$.
  We say that $f$ is a \emph{homeomorphism} if it is
  continuous, invertible, and its inverse $f^{-1}$ is 
  continuous as well.
  
  Write $f$ as its components $(f_1,\ldots,f_m)$, where $f_i:V\longrightarrow\R$ for each $i\in\{1,\ldots,m\}$.
  Then $f$ is \emph{rational} if each function $f_i$ is the ratio of two polynomials, with integer coefficients. 
  
  The sets $V$ and $W$ are \emph{rationally equivalent} 
 if there exists a homeomorphism
 $f : V \longrightarrow W$ 
  such that both $f$ and $f^{-1}$ are rational functions.
  In that case, we write $V\simeq W$.
  \end{definition}


 Be aware that the term \linearExtension, that we define below, has
 other meanings in other contexts.

  \begin{definition}[\linearExtension]
 Given two sets $V\subseteq \R^n$ and $W\subseteq \R^m$, we say that $W$ is a
 \emph{\linearExtension} of $V$ if
 there is an orthogonal projection $\pi:W\longrightarrow V$
 and two vectors $a,b\in \Q^{n}$ 
 such that the mapping 
 \[x \mapsto a\cdot \pi (x) + b\]
 is a continuous bijection.
 In this case we write $V\leqlin W$.
  \end{definition}

\begin{remark}
Rational equivalence linear extension are both transitive relations, i.e., if $V\simeq W$ and $W\simeq U$, then $V\simeq U$, and
if $V\leqlin W$ and $W\leqlin U$, then $V\leqlin U$.

   Furthermore, for a compact set $V$, if $V\simeq W$ or $V\leqlin W$, then $W$ is compact as well.
\end{remark}

  
  
  
  \begin{example}
  Let 
  $V \mydef [1,2]$ and $W \mydef  \{(x,y): x\in [1,2], x = y^2 \}$.
  Then $W$ is not a linear extension of $V$, and  $V$ and $W$
  are neither rationally equivalent, as $W$ has two connected components.
  \end{example}
  
  
 %
 %
 
\subsection{Naming the Variables}
\label{sec:Notation}
We typically denote a new variable with a multi-character symbol such as $\const{v}$.
Here, $v$ is an expression involving already known quantities, and $\const{v}$ should be thought of as a placeholder with that value.
In the case that $v$ is a constant, e.g., $v=2$ such that the variable is $\const{2}$, then we say that $\const{v}$ is a \emph{constant variable}.
By \emph{constructing} a constant variable $\const{v}$ in an \etr formula $\Phi$, we mean introducing variables and constraints to $\Phi$ such that it follows that in every solution to $\Phi$, we have $\const{v}=v$.

The expression $v$ of a variable $\const{v}$ can also involve other parameters such as for instance $\const{x^2+2x+1}$.
In that case $\const{x^2+2x+1}$ should be thought of as a variable holding the value $x^2+2x+1$, where $x$ is another variable or a parameter of the problem.
It should follow from the assumptions or constraints introduced in the concrete case that $\const{x^2+2x+1}$ indeed has the value $x^2+2x+1$ in any solution to $\Phi$.

\section{Reduction to Conjunctive Form}
\label{sec:Conjunction}

\begin{definition}
An \etrconjunction formula $\Phi\mydef\Phi(x_1,\ldots,x_n)$ is a conjunction $C_1\land\ldots\land C_m$,
where $m\geq 0$ and each $C_i$ is of one of the two forms
\begin{align*}
x\geq 0,\quad p(y_1,\ldots,y_l) = 0
\end{align*}
for $x,y_1,\ldots,y_l \in \{x_1, \ldots, x_n\}$ and $p$ is a polynomial.
\end{definition}

Note that since there are no strict inequalities in a formula $\Phi$ in \etrconjunction, the set $V(\Phi)$ is closed.

We show how to reduce a general \ETR formula
to an \etrconjunction formula.
The reduction preserves rational equivalence and runs in linear time.
A similar reduction has been described by Schaefer and \v{S}tefankovi\v{c}~\cite{DBLP:journals/mst/SchaeferS17}.

\begin{lemmaAlph}
\label{lem:Reduction-ETR-Conjunction}
  Given an \etr formula $\Phi$, we can in $O(|\Phi|)$ time compute an \etrconjunction formula $\Psi$ such that $V(\Phi) \simeq V(\Psi)$ and $V(\Phi) \leqlin V(\Psi)$.
\end{lemmaAlph}
\begin{proof}
 We start with an \etr formula $\Phi$
 and modify it repeatedly to attain an \etrconjunction formula
 $\Psi$. 
 Each modification leads
 to an equivalent formula.
 Our modifications can be summarized in four steps.
 (1) Delete ``$\lnot$''. 
 (2) Delete ``$>$''.
 (3) Move ``$\geq$'' to variables only.
 (4) Delete ``$\lor$''.
%
 In the rest of this proof $p$ and $q$ denote polynomials.
 
 Step (1): Here, we merely ''pull`` every negation $\lnot$
 in front of every atomic predicate. 
 For instance 
 $\lnot(A\lor B\lor C)$ becomes 
 $(\lnot A\land \lnot  B\land \lnot C)$.
 To see that this can be done in linear time,
 note that the length of $\Phi$ is at least
 the number of atomic predicates.
 In the end of this process, every atomic predicate
 is preceded by either a negation or not. 
 It may be that $\land$ and $\lor$ symbols are swapped,
 but we count both as one symbol. 
 
 Thereafter each 
 atomic predicate preceded by $\lnot$ is replaced as follows:
 \begin{align*}
  \lnot (q > 0) \quad &\mapsto \quad -q \geq 0\\
  \lnot (q = 0) \quad &\mapsto \quad (q > 0) \,  \lor \,  (-q > 0) \\
  \lnot (q \geq 0)\quad  &\mapsto \quad -q > 0
 \end{align*}
 Those replacements are done repeatedly until there
 are no occurrence of ``$\lnot$'' left in the formula.
 
 Step (2): We replace each 
 inequality as follows:
 \[q > 0 \quad \mapsto  \quad (q\cdot z-1 = 0) \ldots \land z\geq 0.\]
 The dots indicate that the predicate $z\geq 0$ does not 
 immediately follow after $(q\cdot z-1 = 0)$, but will be 
 adjoined at the end of the new formula.
 Furthermore, $z$ denotes a new variable.
 Those replacements are done repeatedly till there
 are no occurrence of ``$>$'' left in the formula.
 
 Step (3): We replace all atomic predicates
 of the form 
 $q \geq 0$
 by the predicate
 $q - z = 0$ and adjoin a new predicate
 $z\geq 0$ at the the end of the formula.
 Again $z$ denotes a new variable.

 Step~(4): We delete disjunctions as follows.
 It will also be necessary to replace some conjunctions.
 Let $\Phi$ be the formula
 after Step~(1)--(3).
 Suppose that there is a disjunction somewhere in $\Phi$, and write it as $\Phi_1\lor\Phi_2$ for two sub-formulas $\Phi_1$ and $\Phi_2$.
 Note that $\Phi_1\lor\Phi_2$ might just be a small part of $\Phi$ -- there will in general be more of $\Phi$ to the right and left of this part.
 
 We want to reduce each of $\Phi_i$ to a single polynomial equation, as follows.
 Note that since we have already performed Step~(1)--(3), there are no inequalities in $\Phi_i$.
 Suppose that $\Phi_i$ is not already a single polynomial equation.
 Then there must somewhere in $\Phi_i$ be either (i) a disjunction $p=0\lor q=0$ or (ii) a conjunction $p =0 \land q = 0$.
 We now explain how to reduce each of these cases to a simpler case.
 
 \begin{itemize}
  \item \textbf{Case (i):}
 We make the replacement
  \[ p =0 \lor q = 0 \quad   \mapsto \quad p\cdot q=0.\]

 \item \textbf{Case (ii):}
 We make the replacement
 \[ p =0 \land q = 0 \quad  \mapsto \quad 
 x\cdot x +  y\cdot y = 0 \ldots 
 \land p -x = 0 \, \land\,  q -y = 0  .\]
 Here, $x$ and $y$ are new variables.
 As in Step~(2), the part following the dots is appended at the end of the complete formula~$\Phi$.
   \end{itemize}
 
 Eventually, we have reduced each $\Phi_i$ to a single polynomial equation.
 Thus the original disjunction $\Phi_1\lor \Phi_2$ has the form as in Case~(i), and we apply the replacement rule described there.
 
 At first, it might seem easier in Case~(ii) 
 to replace
 $p =0 \land q = 0$ by $p\cdot p+q\cdot q = 0$.
 However, we want our reduction to be linear and
 the simplified step could, if done repeatedly,
 lead to very long formulas.
 
 With the replacement rules we have suggested, each iteration reduces the number of 
 disjunctions and conjunctions by one and increases the length of the formula by at most a constant.
 Those replacements are done repeatedly till there
 are no disjunctions left in the formula.

 This reduction takes linear time and the
 final formula $\Psi$ is in conjunctive form.
 We need to describe a rational function
 \[f : V(\Phi) \rightarrow V(\Psi).\]
 Note that $\Psi$ has all the original variables $x_1,\ldots,x_n$
 of $\Phi$ plus some additional variables,
 which we denote by $z_1,\ldots,z_k$.
 If $z\in \{z_1,\ldots,z_k\}$ is introduced
 in step (2), it is assigned the value
 $z = \frac{1}{q}$ and if it is introduced in 
 step (3) or (4), it is assigned the value $z = q$ for some polynomial $q$.
 This defines $f$. Assume that $\Psi$ has
 the free variables $x_1,\ldots,x_n,z_1,\ldots,z_k$,
 where $z_1,\ldots,z_k$ are the variables introduced by the reduction.
 Then \[f^{-1} : (x_1,\ldots,x_n,z_1,\ldots,z_k) \ \mapsto \ (x_1,\ldots,x_n).\]
 Thus $f$ and $f^{-1}$ are
 rational bijective functions.
 Thus $f$ is a homeomorphism.
 The description of $f^{-1}$ implies that 
 $V(\Psi)$ is a linear extension of $V(\Phi)$.
 
\end{proof}

\begin{remark}
 Note that the standard way to remove strict inequalities
 is \[q>0 \mapsto q\cdot y\cdot y -1 = 0.\]
 However, this implies that $y = \pm \sqrt{1/q}$.
 This transformation has two issues. 
 First, the number of solutions in a sense doubles, as the sign of $y$ is not fixed.
 Second, irrational solutions are introduced
 where before may have been only rational solutions.
\end{remark}

\section{Reduction to Compact Semi-Algebraic Sets}
\label{sec:CompactSets}
\begin{definition}
In the problem \etrcompact, we are given an \etrconjunction formula $\Phi$ with the promise that $V(\Phi)$ is compact.
The goal is to decide if $V(\Phi)$ is non-empty.
\end{definition}

In this section, we describe a reduction from 
\etrconjunction to \etrcompact.
We need a tool from 
real algebraic geometry.
The following corollary has been pointed out 
by Schaefer and \v{S}tefankovi\v{c}~\cite{DBLP:journals/mst/SchaeferS17} in a simplified form.

\begin{corollary}
[Basu, Roy~\cite{basu2010bounding} Theorem~2]
\label{cor:BallIntersect}
Let $\Phi$ be an instance of \etr of complexity $L\geq 4$ such that $V(\Phi)$ is a non-empty subset of $\R^n$.
Let $B$ be the set of points in $\R^n$ at distance at most $2^{L^{8n}} = 2^{2^{8n\log L}}$ from the origin.
Then $B\cap V(\Phi)\neq\emptyset$.
\end{corollary}


\begin{lemmaAlph}
\label{lem:Reduction-Conjunction-Compact}
Given an \etrconjunction formula $\Phi\mydef\Phi(x_1,\ldots,x_n)$, we can in $O(|\Phi|+n\log|\Phi|)$ time create an \etrconjunction formula $\Psi$ such that $V(\Psi)$ is compact and $V(\Phi) \neq\emptyset\Leftrightarrow V(\Psi)\neq\emptyset$.
In other words, there is a reduction from \etrconjunction to \etrcompact in near-linear time.
\end{lemmaAlph}
\begin{proof}
 Let an instance $\Phi$ of \etrconjunction be given and define $k \mydef \lceil 8n\log L\rceil$.
 To make an equivalent formula $\Psi$ such that $V(\Psi)$ is compact, we start by including all the variables and constraints of $\Phi$ in $\Psi$.
We then construct a large constant variable $\const{2^{2^k}}$ using \emph{exponentiation by squaring}.
 \begin{align*}
  \const{2^{2^0}} -1-1 &= 0 \\
  \const{2^{2^1}} - \const{2^{2^0}}\cdot \const{2^{2^0}} &= 0 \\
		\vdots & \\
	  \const{2^{2^k}} - \const{2^{2^{k-1}}}\cdot \const{2^{2^{k-1}}} &= 0
 \end{align*}
	For each variable $x$ of $\Phi$, we now introduce the variables $\const{x-2^{2^k}}$ and $\const{x-2^{2^k}-x}$ and the constraints
	\begin{align*}
		\const{x+2^{2^k}}-x-\const{2^{2^k}} & =0 \\
		\const{x+2^{2^k}} & \geq 0 \\
		\const{2^{2^k}-x}-\const{2^{2^k}}+x & =0 \\
		\const{2^{2^k}-x} & \geq 0.
	\end{align*}
Note that this corresponds to introducing the constraint $-2^{2^k}\leq x\leq 2^{2^k}$ in $\Psi$.
 
 It now follows by Corollary~\ref{cor:BallIntersect} that
 $$V(\Phi)\neq\emptyset\Leftrightarrow V(\Psi)\neq\emptyset.$$
 Note that $V(\Psi)$ is compact since $\Psi$ contains no strict inequalities and each variable is bounded.
 This finishes the proof.
\end{proof}

\begin{remark}
 Unfortunately, we do not have $V(\Phi)\simeq V(\Psi)$ in the above reduction.
 That is not possible as it would imply, together with Lemma~\ref{lem:Reduction-ETR-Conjunction}, that an open subset of $\R^n$ is homeomorphic to a compact set.
 We can also not hope for the reduction to yield a linear extension, as a bounded set cannot be a linear extension of an unbounded one.
\end{remark}

\section{Reduction to \etrami}
\label{sec:ami}

\etrami is an abbreviation for
\textit{{\bf E}xitential {\bf T}heory of the {\bf R}eals with
{\bf A}ddition, {\bf M}ultiplication, and {\bf I}nequalities}.

\begin{definition}
\label{def:etrami}
An $\etrami$ formula $\Phi\mydef\Phi(x_1,\ldots,x_n)$ is a conjunction $C_1\land\ldots\land C_m$, where $m\geq 0$ and each $C_i$ is a constraint of one of the forms
$$
 x+y=z,\quad x\cdot y=z, \quad x\geq 0, \quad x=1
$$
for $x,y,z \in \{x_1, \ldots, x_n\}$.
\end{definition}

\begin{lemmaAlph}[\etrami Reduction]
\label{lem:Reduction-AMI-INV}
Given an instance of \etrcompact defined by a formula $\Phi$, we can in $O(|\Phi|)$ time construct
 an \etrami formula $\Psi$ such that
 $V(\Phi) \simeq V(\Psi)$
 and 
 $V(\Phi) \leqlin V(\Psi)$.
\end{lemmaAlph}
\begin{proof}
 Recall that $\Phi$ is a conjunction of atomic formulas of the form $p=0$ for a polynomial $p$ and $x\geq 0$ for a variable $x$.
 Each polynomial $p$ may contain minuses, zeros, and ones.
 The reduction has four steps.
 In each step, we make changes to $\Phi$.
 In the end, $\Phi$ has become a formula $\Psi$ with the desired properties.
 In step~(1)--(3), we remove unwanted ones, zeros and minuses by replacing them by constants.
 In step~(4), we eliminate complicated polynomials.

 Step~(1):
 We introduce the constant variable $\const{1}$ and the constraint $\const{1}=1$ to~$\Phi$.
 We then replace all appearances of $1$ with $\const{1}$ in the atomic formulas of the form $p=0$.
 
 Step~(2):
 We introduce the constant variable $\const{0}$ and the constraint $\const{1}+\const{0}=\const{1}$ to~$\Phi$.
 We then replace all appearances of $0$ with $\const{0}$ except in the constraints of the form $x\geq 0$.
 
 Step~(3):
 We introduce the constant variable $\const{-1}$ and the constraint $\const{1}+\const{-1}=\const{0}$ to~$\Phi$.
 We then replace all appearances of minus with a multiplication by $\const{-1}$ in $\Phi$.
 
 Step~(4):
 We replace bottom up every occurrence 
 of multiplication and addition by a new variable
 and an extra addition or multiplication constraint,
 which will be adjoined at the end of the formula.
 Here are two examples of such replacements:
 \begin{align*}
  x_1 + x_2\cdot x_4 +x_5+x_6 = \const{0} \quad & \mapsto\quad x_1 + z_1 +x_5+x_6 = \const{0}\ldots\land z_1=x_2\cdot x_4 \\
  x_1 + z_1 +x_5+x_6 = \const{0} \quad & \mapsto\quad z_2 +x_5+x_6 = \const{0}\ldots\land z_2=x_1+ z_1
 \end{align*}
 In this way every atomic predicate is eventually transformed to atomic predicates of \etrami or is of the form $x=\const{0}$.
 In the latter case, we replace $x=\const{0}$ by $x+\const{0}=\const{0}$.
 
 To see that the reduction is linear, note that
 every replacement adds a constant to 
 the length of the formula.
 Furthermore, at most linearly many 
 replacements will be done.
 
 Let us show that this reduction preserves
 rational equivalence and \linearExtension.
 This is trivial for steps~(1)--(3), as these just introduce constants in order to rewrite polynomials without using zeros, ones, and minuses.
 In Step~(4), we repeatedly make one of two types of steps, replacing either a multiplication or an addition.
 Thus it is sufficient to show that one such step preserves all of those properties.
 Consider a step where we go from $\Phi_1$ to
 $\Phi_2$ and $\Phi_1$ has the variables
 $x_1,\ldots,x_n$ and $\Phi_2$ has the variables
 $x_1,\ldots,x_n,z$, with $z = x_i \odot x_j$.
 Here $\odot$ is either multiplication or addition.
 This defines $f$ as 
 \[(x_1,\ldots,x_n) \quad \mapsto \quad (x_1,\ldots,x_n, x_i \odot x_j),\]
 and $f^{-1}$ is defined by 
 \[(x_1,\ldots,x_n,z) \quad \mapsto \quad (x_1,\ldots,x_n).\]
 Both functions are rational and bijective, and $f^{-1}$ is an orthogonal projection.
 This implies both rational equivalence and \linearExtension between $V(\Phi)$ and $V(\Psi)$.
\end{proof}

\section{Reduction to \etrsmall}
\label{sec:SMALL}
Let $\delta \in(0,1)$ be given.
The definition of \etrsmall depends on $\delta$, but we will suppress $\delta$ in the notation to keep it simpler.
\begin{definition}
\label{def:etrsmall}
An \emph{\etrsmall formula}~$\Phi\mydef \Phi(x_1,\ldots,x_n)$ is a conjunction $C_1\land\ldots\land C_m$, where $m\geq 0$ and each $C_i$ is a constraint of one of the forms
\[
 x+y=z,\quad x\cdot y=z, \quad x\geq 0,\quad x=\delta
\]
for $x,y,z \in \{x_1, \ldots, x_n\}$.
We define $\delta(\Phi)\mydef \delta$.

In the \emph{\etrsmall problem}, we are given an \etrsmall formula $\Phi$ and promised that $V(\Phi) \subset [-\delta,\delta]^n$.
The goal is to decide whether $V(\Phi) \neq \emptyset$.
\end{definition}

We are going to present a reduction from the problem \etrami to \etrsmall.
As a preparation, we present another tool from
real algebraic geometry.
Schaefer~\cite{schaefer2013realizability} made the following simplification of a result from~\cite{basu2010bounding}, which we will use.
More refined statements can be found in~\cite{basu2010bounding}.

\begin{corollary}[\cite{basu2010bounding}]
\label{cor:BallContain}
If a bounded semi-algebraic set in $\R^n$ has complexity at most
$L \geq 5n$, then all its points have distance at most $2^{2^{L+5}}$
from the origin.
\end{corollary}

\begin{lemmaAlph}[\etrsmall Reduction]
\label{lem:Reduction-AMI-SMALL}
 %
  Given an \etrami formula $\Phi$ such that $V(\Phi)$ is compact, we can in $O(|\Phi|)$ time construct an instance of \etrsmall defined by a formula $\Psi$ such that 
 $V(\Phi) \simeq V(\Psi)$
 and 
 $V(\Phi) \leqlin V(\Psi)$.
\end{lemmaAlph}


\begin{proof}
Let $\Phi$ be an instance of $\etrami$ with $n$ variables $x_1,\ldots,x_n$.
We construct an instance $\Psi$ of $\etrsmall$. 

We set $\eps \mydef \delta\cdot 2^{-2^{L+5}}$, where $L\mydef|\Phi|$. 
In $\Psi$, we first define a constant variable $\const{\eps}$. 
This is obtained by exponentiation by squaring, using $O(L)$ 
new constant variables and constraints.
We first define $\const{\delta}$, $\const{0}$, and $\const{\delta\cdot 2^{-2^0}}$ by the equations
\begin{align*}
\const{\delta} & =\delta \\
\const{0}+\const{\delta} & = \const{\delta} \\
\const{\delta\cdot 2^{-2^0}}+\const{\delta\cdot 2^{-2^0}} & =\const{\delta}.
\end{align*}
We then use the following equations for all $i\in\{0,\ldots,L+4\}$,
\begin{align*}
\const{\delta\cdot 2^{-2^i}}\cdot\const{\delta\cdot 2^{-2^i}} & =\const{\delta^2\cdot 2^{-2^{i+1}}} \\
\const{\delta\cdot 2^{-2^{i+1}}}\cdot \const{\delta} & = \const{\delta^2\cdot 2^{-2^{i+1}}}.
\end{align*}
Finally, we define $\const{\eps}$ by the constraint $\const{\eps}+\const{0}=\const{\delta\cdot 2^{-2^{L+5}}}$.

In $\Psi$, we use the variables $\const{\eps x_1},\ldots,\const{\eps x_n}$ instead of
$x_1,\ldots,x_n$.
An equation of $\Phi$ of the form $x=1$ is transformed to the equation
$\const{\eps x} + \const{0}=\const{\eps}$ in $\Psi$.
An equation of $\Phi$ of the form $x+y=z$ is transformed to the equation
$\const{\eps x} + \const{\eps y} = \const{\eps z}$ of $\Psi$.
For an equation of $\Phi$ of the form $x\cdot y=z$, we also introduce a variable $\const{\eps^2 z}$ of $\Psi$ and the equations
\begin{align*}
\const{\eps x} \cdot \const{\eps y} & = \const{\eps^2 z} \\
\const{\eps}\cdot \const{\eps z} & = \const{\eps^2 z}.
\end{align*}
At last, constraints of the form 
$x \geq 0$ become $\const{\eps x} \geq 0$.

We now describe a function $f : V(\Phi) \rightarrow V(\Psi)$ in order to show that $\Psi$ has the properties stated in the lemma.
Let $\mathbf x\mydef (x_1,\ldots,x_n)\in V(\Phi)$.
In order to define $f$, it suffices to specify the values of the variables of $\Psi$ depending on $\mathbf x$.
For all the constant variables $\const{c}$, we define $\const{c}\mydef c$.
Note that these are all \emph{rational} constants.
For all $i\in\{1,\ldots,n\}$, we now define $\const{\eps x_i} \mydef \eps x_i$ and (when $\const{\eps^2 x_i}$ appears in $\Psi$) $\const{\eps^2 x_i} = \eps^2 x_i$.
Since $\mathbf x$ is a solution to $\Phi$, it follows from the constraints of $\Psi$ that these assignments are a solution to $\Psi$.

We need to verify that $\Psi$ defines an \etrsmall problem, i.e., that $\Psi$ satisfies the promise that $V(\Psi)\subset [-\delta,\delta]^m$, where $m$ is the number of variables of $\Psi$.
To this end, consider an assigment of the variables of $\Psi$ that satisfies all the constraints.
Note first that the constant variables are non-negative and at most $\delta$.
For the other variables, we consider the inverse $f^{-1}$, which is given by the assignment $x_i\mydef \const{\eps x_i}/\eps$ for all $i\in\{1,\ldots,n\}$.
It follows that this yields a solution to~$\Phi$.
Since $V(\Phi)$ is compact, it follows from Corollary~\ref{cor:BallContain} that $|\const{\eps x_i}/\eps|\leq 2^{2^{L+5}}$.
Hence $|\const{\eps x_i}|\leq \eps\cdot 2^{2^{L+5}}=\delta\cdot 2^{-2^{L+5}}\cdot 2^{2^{L+5}}=\delta$.
Similarly, when $\const{\eps^2 x_i}$ is a variable of $\Psi$, we get $|\const{\eps^2 x_i}|\leq \eps\cdot\delta<\delta$.

By the definitions of $f$ and $f^{-1}$, we have now established that $V(\Phi)\simeq V(\Psi)$ and $V(\Phi)\leqlin V(\Psi)$.
The length of $\Psi$ is $O(L)$ longer than the length of $\Phi$, and $\Psi$ can thus be computed in $O(|\Phi|)$ time.
\end{proof}

\section{Reduction to \etrshift}
\label{sec:SHIFT}
Let $\delta > 0$ be given.
The definition of \etrshift depends on $\delta$, but we will suppress $\delta$ in the notation to keep it simpler.
\begin{definition}
\label{def:etrshift}
An \emph{\etrshift formula}~$\Phi\mydef \Phi(x_1,\ldots,x_n)$ is a conjunction
\[
\left(\bigwedge_{i=1}^n 1/2\leq x_i\leq 2\right) \land \left(\bigwedge_{i=1}^m C_i\right),
\]
where $m\geq 0$ and each $C_i$ is of one of the forms
\begin{align*}
 x+y=z,\quad x\cdot y=z 
\end{align*}
for $x,y,z \in \{x_1, \ldots, x_n\}$. 

An instance $\mathcal I\mydef \left[\Phi,(I(x_1),\ldots,I(x_n))\right]$ of the \emph{\etrshift problem} is an \etrshift formula $\Phi$ and, for each variable $x\in\{x_1,\ldots,x_n\}$, an interval $I(x)\subseteq [1/2,2]$  such that $|I(x)|\leq 2 \delta$.
For every \emph{multiplication constraint} $x\cdot y=z$, we have $I(x)\subset [1-\delta,1+\delta]$ and $I(y)\subset [1-\delta,1+\delta]$.
Define $\delta(\mathcal I)\mydef\delta$ and $\Phi(\mathcal I)\mydef\Phi$.

We are promised that $V(\Phi)\subset I(x_1)\times\cdots\times I(x_n)$.
The goal is to decide whether $V(\Phi)\neq\emptyset$.
\end{definition}


We will now present a reduction from the problem \etrsmall to \etrshift.
The following technical lemma is a handy tool to show that all variables $x$ of the constructed \etrshift problem are in the ranges $I(x)$ we are going to specify.

\begin{restatable}{lemma}{CoreValue}
\label{lem:Core-Value}
 Let $g(x,y) \mydef  \frac{p(x,y)}{q(x,y)}$ be a rational function such that
\begin{align*}
p(x,y) & \mydef a_1 \ x^2 + a_2 \ xy + a_3 \ y^2  + a_4 \ x + a_5 \ y + a_6,\quad\text{and} \\
q(x,y) & \mydef b_1 \ x^2 + b_2 \ xy + b_3 \ y^2  + b_4 \ x + b_5 \ y + b_6,
\end{align*}
where $b_6>0$.
Let $\alpha\mydef |b_1|+\ldots+|b_5|$ and $\beta\mydef |b_1|+\ldots+|b_5|$.

Then for all $x,y\in [-\delta,\delta]$, where $\delta\in[0,1]$, we have
\begin{align}\label{lem:Core-Value:bound}
g(x,y)\in\left[\frac{-\alpha\delta+a_6}{\beta \delta+b_6},\frac{\alpha\delta+a_6}{-\beta\delta+b_6}\right].
\end{align}

In particular,
\begin{itemize}
\item[(a)]
if $q(x,y)=b_6=1$ and $a_1,\ldots,a_5\in [0,c]$ for some $c\geq 0$, then
\[
g(x,y)\in[a_6-5c\delta,a_6+5c\delta],
\]
and

\item[(b)]
if $a_1,\ldots,a_5,b_1,\ldots,b_5\in[-c, c]$ for some $c\geq 0$ and $\delta\leq\frac{\eps b_6^2}{5c(a_6+(1+\eps)b_6)}$ for a given $\eps>0$, then
\[
g(x,y)\in [a_6/b_6-\eps,a_6/b_6+\eps].
\]
\end{itemize}
\end{restatable}

\begin{proof}
We bound each term in each polynomial from below and above and get
\begin{align*}
p(x,y) \in [ & -(|a_1|+|a_2|+|a_3|)\delta^2-(|a_4|+|a_5|)\delta+a_6,\\
 & + (|a_1|+|a_2|+|a_3|)\delta^2+(|a_4|+|a_5|)\delta+a_6], \quad\text{and} \\
q(x,y) \in [ & -(|b_1|+|b_2|+|b_3|)\delta^2-(|b_4|+|b_5|)\delta+b_6, \\
 & + (|b_1|+|b_2|+|b_3|)\delta^2+(|b_4|+|b_5|)\delta+b_6].
\end{align*}

The bounds~\eqref{lem:Core-Value:bound} now follows as $\delta\in[0,1]$ so that $\delta^2\leq\delta$.

For part (a), note that $\beta=0$ and $\alpha\in [0,5c]$.
For part (b), we get that
\[
g(x,y)\in\left[\frac{-5c\delta+a_6}{5c\delta+b_6},\frac{5c\delta+a_6}{-5c\delta+b_6}\right].
\]
One can then check that if $\delta\leq\frac{\eps b_6^2}{5c(a_6+(1+\eps)b_6)}$, that range is contained in $[a_6/b_6-\eps,a_6/b_6+\eps]$.
\end{proof}

\begin{lemmaAlph}[\etrshift Reduction]
\label{lem:Reduction-SMALL-SHIFT}
Let $\delta_2\in(0,1/4)$ be given, and let $\delta_1\leq \delta_2/5$ such that $\delta_1=2^{-l}$ for $l\in\N$.
Consider an instance of the \etrsmall problem, defined by a formula $\Phi_1$ such that $\delta(\Phi_1)=\delta_1$.
We can in $O(|\Phi_1| +l)$ time compute an instance $\mathcal I_2$ of the \etrshift problem with $\delta(\mathcal I_2)=\delta_2$ and formula $\Phi_2\mydef\Phi(\mathcal I_2)$ such that
 $V(\Phi_1) \simeq V(\Phi_2)$ and
 $V(\Phi_1) \leqlin V(\Phi_2)$.
\end{lemmaAlph}

\begin{proof}
In the following, we specify the variables and constraints we add to $\Phi_2$.
Define $\Delta\mydef 1-\delta_1$.
As a first step, we construct constant variables $\const{c}$ for each of $c\in\{1/2, 3/4, 1, 3/2\}$, as follows.
We first use the constraint $\const{1}\cdot\const{1}=\const{1}$.
Note that the solutions to this are $\const{1}\in\{0,1\}$, but since we are restricted to $[1/2,2]$, we conclude that $\const{1}=1$.
We observe that $\const{1}$ is in the promised range $[1-\delta_2,1+\delta_2]$ of variables involved in multiplication constraints.

We can then construct the other constants as follows.
\begin{align*}
\const{1/2} + \const{1/2} & = \const{1} \\
\const{1}+\const{1/2} & = \const{3/2} \\
\const{3/4} + \const{3/4} & = \const{3/2}
\end{align*}

We now show how to construct a constant variable $\const{\Delta}$.
To this end, we construct constant variables $\const{1-2^{-i}}$ for $i\in\{1,\ldots,l\}$, so that $\const{\Delta}$ is a synonym for $\const{1-2^{-l}}$.
For the base case $i=1$, note that this is just the already constructed $\const{1/2}$.
Suppose inductively that we have constructed the constant variable $\const{1- 2^{-i}}$.
In order to construct $\const{1-2^{-(i+1)}}$, we proceed as follows.
\begin{align*}
\const{1- 2^{-i}} + \const{1} & = \const{2- 2^{-i}} \\
\const{1-2^{-(i+1)}}+\const{1-2^{-(i+1)}} & = \const{2- 2^{-i}}
\end{align*}
Thus we can generate a variable with the value $\const{\Delta}$ in $O(l)$ steps.

For each of the constant variables $\const{c}$ thus created, we define $I(\const{c})\mydef [c-\delta_2,c+\delta_2]\cap[1/2,2]$.
Note that it follows from the constraints that in any solution to $\Phi_2$, we must have $\const{c}=c$.

For each each variable $x \in [-\delta_1,\delta_1]$ of $\Phi_1$, we make a corresponding variable $\const{x+1}$ of $\Phi_2$.
As we shall see, for every solution $\mathbf x\mydef (x_1,\ldots,x_n)$ of $\Phi_1$, there will be a corresponding solution to $\Phi_2$ with $\const{x_i+1}=x_i+1$, and vice versa.

For each variable $x$ of $\Phi_1$, we construct the variables $\const{x+3/2}$, $\const{x+3/4}$, and $\const{x+\Delta}$ as follows.
\begin{align*}
\const{x+1} + \const{1/2} & = \const{x+3/2}\\
\const{x+3/4}+\const{3/4} & = \const{x+3/2}\\
\const{x+3/4} + \const{\Delta} & = \const{x+3/4+\Delta} \\
\const{x+\Delta}+\const{3/4} & = \const{x+3/4+\Delta}.
\end{align*}
For each of these of the form $\const{x+b}$, $b\in\{3/4,\Delta,3/2\}$, it holds that if $\const{x+1}=x+1$, then $\const{x+b}=x+b$.

We now go through the constraints of $\Phi_1$ and create equivalent constraints in $\Phi_2$.
For each equation $x=\delta_1$ of $\Phi_1$, we add the equation $\const{x+\Delta}=1$ to $\Phi_2$.
The equation implies that if $\const{x+1}=x+1$, then $x=\delta_1$.

For each equation $x+y=z$ of $\Phi_1$, we add
\begin{align}\label{etrshift:addition}
\const{x+3/4} + \const{y+3/4} = \const{z+3/2}
\end{align}
to $\Phi_2$.
This equation implies that if $\const{x+1}=x+1$ and $\const{y+1}=y+1$, then $\const{z+1}=x+y+1$.

For each equation $x \cdot y=z$ of $\Phi_1$, we have the following set of equations in $\Phi_2$.
\begin{align*}
\const{x+1} \cdot \const{y+1} &= \const{xy + x+y+1} \\
\const{xy + x + y + 1}+\const{1/2} & =\const{xy + x+y+3/2} \\
\const{xy+x+3/4} + \const{y+3/4}  &= \const{xy + x + y + 3/2}  \\
\const{xy+x+3/4}+\const{3/4} & = \const{xy+x+3/2} \\
\const{z+3/4} + \const{x+3/4}  &= \const{xy+x+3/2}
\end{align*}
These equations imply that if $\const{x+1}=x+1$ and $\const{y+1}=y+1$, then $\const{z+1}=xy+1$.

At last, for each constraint $x\geq 0$ of $\Phi_1$, we introduce the variable $\const{x+1/2}$ of $\Phi_2$ and the equation $\const{x+1/2} + \const{1/2} = \const{x+1}$.
The constraint $\const{x+1/2} \geq 1/2$, which holds for all variables of $\Phi_2$ by definition of \etrshift, then corresponds to $x \geq 0$.

Note that each of the introduced variables has the form $\const{p(x,y)}$, where $p(x,y)$ is a polynomial of degree at most $2$ and with constant term $c\mydef p(0,0)\in \{1/2,3/4,1,3/2\}$.
We now define $I(\const{p(x,y)})\mydef [c-\delta_2,c+\delta_2]\cap[1/2,2]$.

The construction of $\Phi_2$ is now finished, and we need to check that it has the claimed properties.
Let the variables of $\Phi_1$ be $x_1,\ldots,x_n$ and the variables of $\Phi_2$ be
$$\const{x_1+1},\ldots,\const{x_n+1},\const{y_1},\ldots,\const{y_m}.$$
For each variable $\const{y_i}$, $i\in\{1,\ldots,m\}$, the expression $y_i$ is a polynomial of degree at most two in two variables among $x_1,\ldots,x_n$ (this includes the case that $y_i$ is a constant).
Consider any solution $\mathbf x\mydef (x_1,\ldots,x_n)\in [-\delta_1,\delta_1]^n$ to $\Phi_1$.
We get a corresponding solution $f(\mathbf x)$ to $\Phi_2$ as follows.
We set $\const{x_i+1}\mydef x_i+1$ for every $i\in\{1,\ldots,n\}$.
For every $i\in\{1,\ldots,m\}$, $y_i$ is a (possibly constant) polynomial in two variables among $x_1,\ldots,x_n$, and we assign $\const{y_i}$ the value we get when evaluating this polynomial.

In order to show that this yields a solution to $\Phi_2$, we consider the constraint~\eqref{etrshift:addition} introduced to $\Phi_2$ due to an addition $x+y=z$ of $\Phi_1$.
The other constraints can be verified in a similar way.
Due to the construction of $\const{x+3/4}$, it follows from $\const{x+1}\mydef x+1$ that $\const{x+3/4}=x+3/4$, and similarly that $\const{y+3/4}=y+3/4$ and $\const{z+3/2}=z+3/2$.
Hence we have
\begin{align*}
\const{x+3/4}+\const{y+3/4}=x+3/4+y+3/4=z+3/2=\const{z+3/2},
\end{align*}
so indeed the constraint is satisfied.

Note that the inverse of $f$ is
\[
f^{-1} : (\const{x_1+1},\ldots,\const{x_n+1},\const{y_1},\ldots,\const{y_m})\mapsto (\const{x_1+1}-1,\ldots,\const{x_n+1}-1).
\]
We now show that $f^{-1}$ is a map from $V(\Phi_2)$ to $V(\Phi_1)$, i.e., that given any solution to $\Phi_2$, $f^{-1}$ yields a solution to $\Phi_1$.
Consider a constraint of $\Phi_1$ of the form $x+y=z$.
We then have
\begin{align*}
x+y & =\const{x+1}-1+\const{y+1}-1=\const{x+3/4}+1/4-1+\const{y+3/4}+1/4-1 \\
& =\const{z+3/2}-3/2=z.
\end{align*}
In a similar way, the other constraints of $\Phi_1$ can be shown to hold due to the constraints of $\Phi_2$.

It follows that $f$ is a rational homeomorphism so $V(\Phi_1)\simeq V(\Phi_2)$, and since $f^{-1}$ merely subtracts $1$ from some variables, we also have $V(\Phi_1)\leqlin V(\Phi_2)$.

At last, we need to verify that $\Phi_2$ satisfies the promise that in every solution, each variable $\const{p(x,y)}$ is in the interval $I(\const{p(x,y)})$.
Here, $p(x,y)$ is a polynomial of degree at most $2$ and with constant term $c\mydef p(0,0)\in \{1/2,3/4,1,3/2\}$.
By the map $f^{-1}$, we get a solution to $\Phi_1$ by the assignments $x\mydef \const{x+1}-1$ and $y\mydef \const{y+1}-1$ (and similarly for the remaining variables of $\Phi_1$).
It then follows from the constraints of $\Phi_2$ that $\const{p(x,y)}=p(x,y)$.
By the promise of $\Phi_1$, we get that $x,y\in[-\delta_1,\delta_1]$.
The coefficients of the non-constant terms of $p(x,y)$ are all either $0$ or $1$.
We therefore get by Lemma~\ref{lem:Core-Value} that since $x,y\in[-\delta_1,\delta_1]$, then $p(x,y)\in[c-5\delta_1,c+5\delta_1]\subset[c-\delta_2,c+\delta_2]$.

Recall that $I(\const{p(x,y)})\mydef [c-\delta_2,c+\delta_2]\cap[1/2,2]$ and that $\delta_2< 1/4$.
With the exception of the case $c=1/2$, we therefore have that $I(\const{p(x,y)})=[c-\delta_2,c+\delta_2]$, so it follows that $\const{p(x,y)}\in I(\const{p(x,y)})$.
Note that the case $c=1/2$ only occurs when $p(x,y)=x+1/2$ and $I(\const{x+1/2})= [1/2,1/2+\delta_2]$.
In this case, there is a constraint $x\geq 0$ in $\Phi_1$.
Hence $x\in [0,\delta_1]$, so that $\const{x+1/2}=x+1/2\in[1/2,1/2+\delta_1]\subset [1/2,1/2+\delta_2]=I(\const{x+1/2})$.

In order to construct $\const{\Delta}$ in $\Phi_2$, we introduce $O(l)$ variables and constraints.
For each of the $O(|\Phi_1|)$ variables and constraints of $\Phi_1$, we make a constant number of variables and constraints in $\Phi_2$.
It thus follows that the running time is $O(|\Phi_1|+l)$.
This completes the proof.
\end{proof}

\section{Reduction to \etrsquare}
\label{sec:SQUARE}
In this and the following section, we show that the problem \etrshift remains essentially equally hard even when we only allow more specialized types of multiplications in our formulas.
In this section, we require every multiplication to be a \emph{squaring} of the form $x^2=y$, and in the following section, we only allow \emph{inversion} of the form $x\cdot y=1$.
The result that these two restricted types of constraints preserve the full expressibility of \etrshift is related to the result of Aho \etal~\cite[Section~8.2]{e_k1974design} that squaring and taking reciprocals of integers require work proportional to that of multiplication.

Let $\delta > 0$ be given.
The definition of \etrsquare depends on $\delta$, but we will suppress $\delta$ in the notation to keep it simpler.
\begin{definition}
\label{def:etrsquare}
An \emph{\etrsquare formula}~$\Phi\mydef \Phi(x_1,\ldots,x_n)$ is a conjunction
\[
\left(\bigwedge_{i=1}^n 1/2\leq x_i\leq 2\right) \land \left(\bigwedge_{i=1}^m C_i\right),
\]
where $m\geq 0$ and each $C_i$ is of one of the forms
\begin{align*}
 x+y=z,\quad x^2=y 
\end{align*}
for $x,y,z \in \{x_1, \ldots, x_n\}$. 

An instance $\mathcal I=\left[\Phi,(I(x_1),\ldots,I(x_n))\right]$ of the \emph{\etrsquare problem} is an \etrsquare formula $\Phi$ and,  for each variable $x\in\{x_1,\ldots,x_n\}$, an interval $I(x)\subseteq [1/2,2]$ such that $|I(x)|\leq 2 \delta$.
For every \emph{squaring constraint} $x^2=y$, we have $I(x)\subset [1-\delta,1+\delta]$.
Define $\delta(\mathcal I)\mydef\delta$ and $\Phi(\mathcal I)\mydef\Phi$.

We are promised that $V(\Phi)\subset I(x_1)\times\cdots\times I(x_n)$.
The goal is to decide whether $V(\Phi)\neq\emptyset$.
\end{definition}
Below, we present a reduction from the problem \etrshift 
to \etrsquare.  
\begin{lemmaAlph}[\etrsquare Reduction]
\label{lem:Reduction-SHIFT-SQUARE}
Let $\delta_2\in(0,1/4)$ be given, and let $\delta_1\mydef \delta_2/10$.
Consider an instance $\mathcal I_1$ of the \etrshift problem such that $\delta(\mathcal I_1)=\delta_1$, and let $\Phi_1\mydef\Phi(\mathcal I_1)$.
We can in $O(|\Phi_1|)$ time compute an instance $\mathcal I_2$ of the \etrsquare problem with $\delta(\mathcal I_2)=\delta_2$ and formula $\Phi_2\mydef\Phi(\mathcal I_2)$ such that
 $V(\Phi_1) \simeq V(\Phi_2)$ and
 $V(\Phi_1) \leqlin V(\Phi_2)$.

\end{lemmaAlph}

\begin{proof}
As in the proof of Lemma~\ref{lem:Reduction-SMALL-SHIFT}, we first construct constant variables $\const{b}$ for each of $b\in\{1/2, 3/4, 1, 3/2\}$.
The only difference is that we now construct $\const{1}$ using the constraint $\const{1}^2=\const{1}$.
The other constants are then constructed in exactly the same way as before.

We include each variable $x$ of $\Phi_1$ in $\Phi_2$ as well, and reuse the interval $I(x)$ from $\mathcal I_1$ in $\mathcal I_2$.
We also reuse all contraints from $\Phi_1$ of the form $a+b=c$ in $\Phi_2$, but we have to do something different for the constraints $a\cdot b=c$.
The idea is very simple and was also used by Aho \etal~\cite[Section~8.2]{e_k1974design}, and can be expressed by the equations
\begin{align*}
 x &= (a+b)/2\\
 y &= (a-b)/2\\
 u &= x^2 \\ 
 v &= y^2 \\
 c &= u-v = a\cdot b.
\end{align*}
If there were no range constraints, we could just replace each multiplication $a\cdot b =c$ of $\Phi_1$ by equations as above (after rewriting the subtractions as additions, etc.).
However, in our situation all 
intermediate variables $w$ need to
be in a range $I(w)\subset [1/2,2]$ in any solution.
While this makes the description more involved, careful shifting will work for us.

Let $a\cdot b =c$ be a multiplication constraint in $\Phi_1$.
Let us rename the variables as $\const{x+1}\mydef a$, $\const{y+1}\mydef b$, and $\const{xy+x+y+1}\mydef c$, so that $a\cdot b =c$ becomes
\begin{align}
\label{etrsquare:mult}
\const{x+1}\cdot\const{y+1}=\const{xy+x+y+1}.\tag{$\dagger$}
\end{align}
Consider two numbers $x,y\in\R$ and the two conditions
\begin{align}
& \const{x+1}=x+1\quad\text{and}\quad \const{y+1}=y+1,\label{etrsquare:cond}\tag{$\star$} \\
& \const{xy+x+y+1} =xy+x+y+1.\label{etrsquare:cond2}\tag{$\star\star$}
\end{align}
We claim that~\eqref{etrsquare:mult} is equivalent to
\begin{align}
\text{\eqref{etrsquare:cond} implies \eqref{etrsquare:cond2}}.\label{etrsquare:cond3}\tag{$\ddagger$}
\end{align}
To show this claim, suppose first that~\eqref{etrsquare:mult} holds.
Define $x\mydef\const{x+1}-1$ and $y\mydef\const{y+1}-1$.
Then $\const{xy+x+y+1}=\const{x+1}\cdot\const{y+1}=(x+1)\cdot (y+1)=xy+x+y+1$, so that~\eqref{etrsquare:cond2} holds.
Hence we have~\eqref{etrsquare:cond3}.
On the other hand, suppose that~\eqref{etrsquare:cond3} holds, and define $x\mydef\const{x+1}-1$ and $y\mydef\const{y+1}-1$ so that~\eqref{etrsquare:cond} holds.
Our assumption implies that~\eqref{etrsquare:cond2} holds, and we thus have $\const{xy+x+y+1}=xy+x+y+1=(x+1)\cdot(y+1)=\const{x+1}\cdot\const{y+1}$, so that~\eqref{etrsquare:mult} holds.
Our aim is therefore to make constraints in $\Phi_2$ that ensure \eqref{etrsquare:cond3}.

For every variable $\const{q}$ of $\Phi_2$, we can construct $\const{q/2}$ using the constraint $\const{q/2} + \const{q/2} = \const{q}$.
Similarly, we can construct $\const{2q}$ by $\const{q} + \const{q} = \const{2q}$.

The construction of $\const{xy+x+y+1}$ is shown in Figure~\ref{fig:VariableSequence2}.
The diagram should be understood in the following way.
 We start with the original variables $\const{x+1}$ and $\const{y+1}$.
 Each arrow is labeled with the operation that leads to the new variable.
It is straightforward to check that the construction ensures condition~\eqref{etrsquare:cond3}.

Note that each of the constructed auxilliary variables has the form $\const{p(x,y)}$, where $p(x,y)$ is a second degree polynomial with constant term $c\mydef p(0,0)\in\{3/4,1,3/2\}$.
We define $I(\const{p(x,y)})\mydef [c-\delta_2,c+\delta_2]$.
This finishes the construction of $\mathcal I_2$.
Note that since $\delta_2\leq 1/4$ and $c\in\{3/4,1,3/2\}$, we get that $I(\const{p(x,y)})\subset[1/2,2]$, as required.

We now verify that $\mathcal I_2$ has the claimed properties.
Consider a solution $\mathbf x\in V(\Phi_1)$.
We point out an equivalent solution to $\Phi_2$.
For all the variables of $\Phi_2$ that also appear in $\Phi_1$, we use the same value as in the solution $\mathbf x$.
Each auxilliary variable has the form $\const{p(x,y)}$, where $p(x,y)$ is a second degree polynomial, and $\Phi_1$ contains the multiplication constraint~\eqref{etrsquare:mult}.
We then get from the promise of $\Phi_1$ that $\const{1+x}=1+x$ and $\const{1+y}=1+y$ for $x,y\in[-\delta_1,\delta_1]$.
From the construction shown in Figure~\ref{fig:VariableSequence2}, it follows that in order to get a solution to $\Phi_2$, we must define $\const{p(x,y)}\mydef p(x,y)$.
Recall that the constant term $c\mydef p(0,0)$ satisfies $c\in \{3/4,1,3/2\}$, and note that all the coefficients of the non-constant terms of $p(x,y)$ are in the interval $[0,2]$.
We then get from Lemma~\ref{lem:Core-Value} that $\const{p(x,y)}\in[c-10\delta_1,c+10\delta_1]=[c-\delta_2,c+\delta_2]=I(\const{p(x,y)})\subset[1/2,2]$.
The variables are thus in the range $[1/2,2]$, so we have described a solution to $\Phi_2$.

Similarly, we see that any solution to $\Phi_2$ corresponds to a solution to $\Phi_1$.
Using the promise of $\mathcal I_1$ and Lemma~\ref{lem:Core-Value} as above, we then also confirm the promise of $\mathcal I_2$ that each variable $u$ of $\Phi_2$ is in $I(u)$.

\begin{figure}[htbp]
 \centering
 \includegraphics[page = 6]{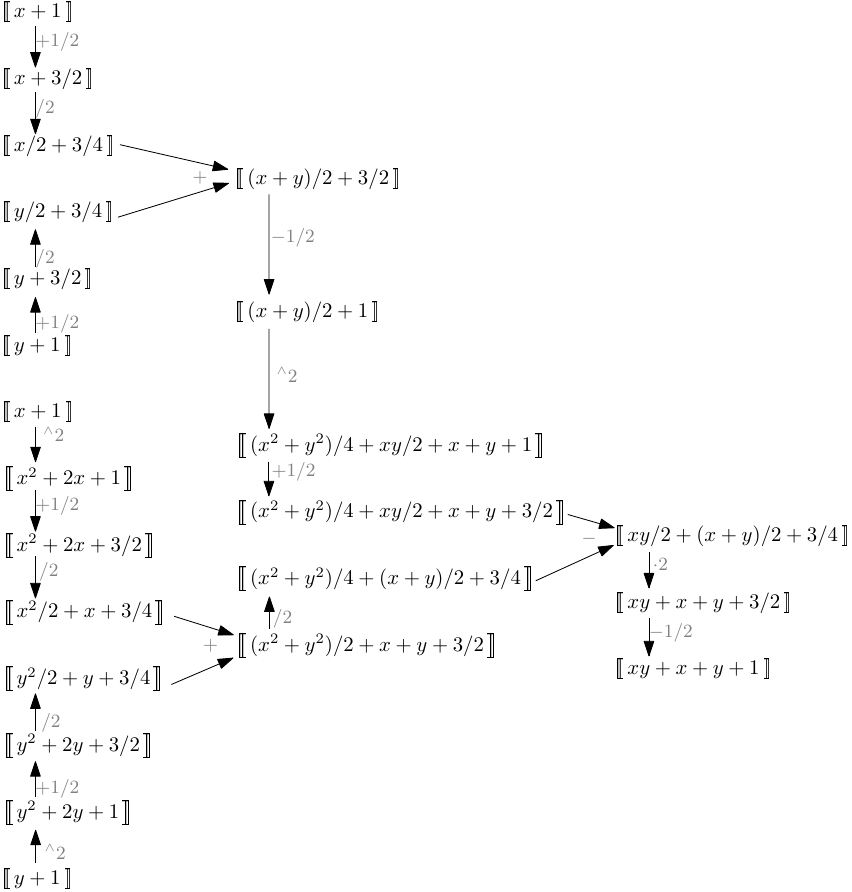}
 \caption{The construction of $\const{xy+x+y+1}$ from $\const{x+1}$ and $\const{y+1}$.
Squaring a variable is denoted by $^\wedge 2$.}
 \label{fig:VariableSequence2}
\end{figure}

The correspondance described implies that $V(\Phi_1)\simeq V(\Phi_2)$ and $V(\Phi_1)\leqlin V(\Phi_2)$.
For each constraint of $\Phi_1$, we introduce $O(1)$ variables and constraints of $\Phi_2$, so the reduction takes $O(|\Phi_1|)$ time.
\end{proof}

\section{Reduction to \etrinv}
\label{sec:INV}
Let $\delta > 0$ be given.
The definition of \etrinv depends on $\delta$, but we will suppress $\delta$ in the notation to keep it simpler.
\begin{definition}
\label{def:etrinv}
An \emph{\etrinv formula}~$\Phi=\Phi(x_1,\ldots,x_n)$ is a conjunction
\[
\left(\bigwedge_{i=1}^n 1/2\leq x_i\leq 2\right) \land \left(\bigwedge_{i=1}^m C_i\right),
\]
where $m\geq 0$ and each $C_i$ is of one of the forms
\begin{align*}
 x+y=z,\quad x\cdot y=1 
\end{align*}
for $x,y,z \in \{x_1, \ldots, x_n\}$. 

An instance $\mathcal I=\left[\Phi,(I(x_1),\ldots,I(x_n))\right]$ of the \emph{\etrinv problem} is an \etrinv formula $\Phi$ and,  for each variable $x\in\{x_1,\ldots,x_n\}$, an interval $I(x)\subseteq [1/2,2]$ such that $|I(x)|\leq 2 \delta$.
Define $\delta(\mathcal I)\mydef\delta$ and $\Phi(\mathcal I)\mydef\Phi$.

We are promised that $V(\Phi)\subset I(x_1)\times\cdots\times I(x_n)$.
The goal is to decide whether $V(\Phi)\neq\emptyset$.
\end{definition}


We will now present a reduction 
from the problem \etrsquare to \etrinv.  
\begin{lemmaAlph}[\etrinv Reduction]
\label{lem:Reduction-SQUARE-INV}
Let $\delta_2\in(0,1/6)$ be given, and let $\delta_1\mydef \delta_2/1800$.
Consider an instance $\mathcal I_1$ of the \etrsquare problem such that $\delta(\mathcal I_1)=\delta_1$, and let $\Phi_1\mydef\Phi(\mathcal I_1)$.
We can in $O(|\Phi_1|)$ time compute an instance $\mathcal I_2$ of the \etrinv problem with $\delta(\mathcal I_2)=\delta_2$ and formula $\Phi_2\mydef\Phi(\mathcal I_2)$ such that
 $V(\Phi_1) \simeq V(\Phi_2)$ and
 $V(\Phi_1) \leqlin V(\Phi_2)$.
\end{lemmaAlph}

\begin{proof}
As in the proof of Lemma~\ref{lem:Reduction-SMALL-SHIFT}, we first construct constant variables $\const{b}$ for each of $b\in\{1/2, 3/4, 1, 3/2\}$.
The only difference is that we now construct $\const{1}$ using the constraint $\const{1}\cdot\const{1}=1$.
It follows from this constraint that $\const{1}\in\{-1,1\}$, and since $\const{1}\in[1/2,2]$, we must have $\const{1}=1$ in every valid solution.
The other constants are then constructed in exactly the same way as before.
For this reduction we also need the constant variable $\const{2/3}$ which is constructed as $\const{2/3}\cdot \const{3/2} = 1$.


We include each variable $x$ of $\Phi_1$ in $\Phi_2$ as well, and reuse the interval $I(x)$ from $\mathcal I_1$ in $\mathcal I_2$.
We also reuse all contraints from $\Phi_1$ of the form $a+b=c$ in $\Phi_2$, but we have to do something different for the squaring constraints $a^2=b$.
In $\Phi_2$, we rename the variables as $\const{x+1}\mydef a$ and $\const{x^2+2x+1}\mydef b$, so that $a^2 =b$ becomes
\begin{align}
\label{etrinv:sqr}
\const{x+1}^2=\const{x^2+2x+1}.\tag{$\dagger$}
\end{align}
Consider a number $x\in\R$ and the two conditions
\begin{align}
& \const{x+1}=x+1,\label{etrinv:cond}\tag{$\star$} \\
& \const{x^2+2x+1}=x^2+2x+1.\label{etrinv:cond2}\tag{$\star\star$}
\end{align}
As in the proof of Lemma~\ref{lem:Reduction-SHIFT-SQUARE}, one can prove that~\eqref{etrinv:sqr} is equivalent to
\begin{align}
\text{\eqref{etrinv:cond} implies \eqref{etrinv:cond2}}.\label{etrinv:cond3}\tag{$\ddagger$}
\end{align}
Our aim is therefore to make constraints in $\Phi_2$ that ensure~\eqref{etrinv:cond3}.

 In the same way as described in Section~\ref{sec:SHIFT}
 and Section~\ref{sec:SQUARE}, we can add values, subtract values,
 half and double variables.
Figure~\ref{fig:VariableSequence3} shows the construction of $\const{x^2+2x+1}$ using these tricks as well as inversions.
It is straightforward to check that the construction ensures condition~\eqref{etrinv:cond3}.

\begin{figure}[htbp]
 \centering
 \includegraphics[page = 7]{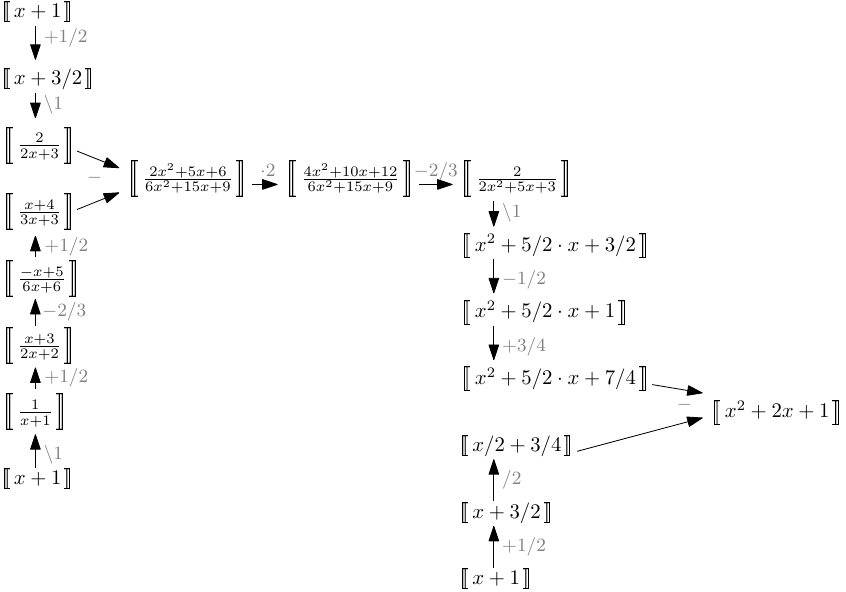}
 \caption{The sequence above enforces $ y = x^2$, using only addition and inversion.
Inversion is denoted by $\setminus 1$.}
 \label{fig:VariableSequence3}
\end{figure}

For all the variables $x$ of $\Phi_2$ that also appear in $\Phi_1$, we define the interval $I(x)$ in $\mathcal I_2$ as it is in $\mathcal I_1$.
Each of the auxilliary variables has the form $\const{\frac{p(x)}{q(x)}}$, where $p(x)$ and $q(x)$ are polynomials of degree $2$.
We define $c\mydef \frac{p(0)}{q(0)}$ and note that $c\in[2/3,7/4]$.
We define $I\left(\const{\frac{p(x)}{q(x)}}\right)\mydef [c-\delta_2,c+\delta_2]$.
This finishes the construction of $\mathcal I_2$.
Note that since $\delta_2\leq 1/6$ and $c\in[2/3,7/4]$, we get that $I\left(\const{\frac{p(x)}{q(x)}}\right)\subset[1/2,2]$.

We now verify that $\mathcal I_2$ has the claimed properties.
Consider a solution $\mathbf x\in V(\Phi_1)$.
For all the variables of $\Phi_2$ that also appear in $\Phi_1$, we use the same value as in the solution $\mathbf x$.
Each auxilliary variable has the form $\const{\frac{p(x)}{q(x)}}$, where $p(x)$ and $q(x)$ are polynomials of degree $2$, and $\Phi_1$ contains the squaring constraint~\eqref{etrinv:sqr}.
We then get from the promise of $\Phi_1$ that $\const{1+x}=1+x$ for $x\in[-\delta_1,\delta_1]$.
From the construction shown in Figure~\ref{fig:VariableSequence3}, it follows in order to get a solution to $\Phi_2$, we must define $\const{\frac{p(x)}{q(x)}}\mydef \frac{p(x)}{q(x)}$.
We are going to apply case~(b) of Lemma~\ref{lem:Core-Value} to show that this solution stays in the required range $[1/2,2]$.
The coefficients of the non-constant terms of $p(x)$ and $q(x)$ are in the range $[-1,15]$.
Denote by $a_6$ and $b_6$ the constant terms of $p(x)$ and $q(x)$, respectively.
We observe that $b_6^2\geq 1$ and $a_6+(1+\delta_2)b_6\leq 6+2\cdot 9=24$.
We therefore get that since $\delta_1\mydef \delta_2/1800=\frac{\delta_2}{5\cdot 15\cdot 24}<\frac{\delta_2 b_6^2}{5\cdot 15(a_6+(1+\delta_2)b_6)}$, then $\const{\frac{p(x)}{q(x)}}\in [c-\delta_2,c+\delta_2]\subset[1/2,2]$.

Similarly, we see that any solution to $\Phi_2$ corresponds to a solution to $\Phi_1$.
Using the promise of $\mathcal I_1$ and Lemma~\ref{lem:Core-Value} as above, we then also confirm the promise of $\mathcal I_2$ that each variable $u$ of $\Phi_2$ is in $I(u)$.

The correspondance described implies that $V(\Phi_1)\simeq V(\Phi_2)$ and $V(\Phi_1)\leqlin V(\Phi_2)$.
For each constraint of $\Phi_1$, we introduce $O(1)$ variables and constraints of $\Phi_2$, so the reduction takes $O(|\Phi_1|)$ time.
\end{proof}

\bibliographystyle{plain}
\bibliography{lib}

\begin{thebibliography}{10}

\bibitem{ARTETR}
Mikkel Abrahamsen, Anna Adamaszek, and Tillmann Miltzow.
\newblock The art gallery problem is $\exists \mathbb{R}$-complete.
\newblock In {\em Proceedings of the 50th Annual {ACM} {SIGACT} Symposium on
  Theory of Computing, {STOC} 2018, Los Angeles, CA, USA, June 25-29, 2018},
  pages 65--73, 2018.

\bibitem{e_k1974design}
Alfred~V. Aho, John~E. Hopcroft, and Jeffrey~D. Ullman.
\newblock {\em The design and analysis of computer algorithms}.
\newblock Addison-Wesley, 1975.

\bibitem{basu2010bounding}
Saugata Basu and Marie-Fran{\c{c}}oise Roy.
\newblock Bounding the radii of balls meeting every connected component of
  semi-algebraic sets.
\newblock {\em Journal of Symbolic Computation}, 45(12):1270--1279, 2010.

\bibitem{blum2012complexity}
Lenore Blum, Felipe Cucker, Michael Shub, and Steve Smale.
\newblock {\em Complexity and real computation}.
\newblock Springer Science \& Business Media, 2012.

\bibitem{canny1988some}
John Canny.
\newblock Some algebraic and geometric computations in {PSPACE}.
\newblock In {\em Proceedings of the twentieth annual ACM Symposium on Theory
  of Computing (STOC 1988)}, pages 460--467. ACM, 1988.

\bibitem{AreasKleist}
Michael~G. Dobbins, Linda Kleist, Tillmann Miltzow, and Pawe\l{} Rz{\c
  a}{\.z}ewski.
\newblock {$\forall \exists \mathbb{R}$-completeness and area-universality}.
\newblock {\em Arxiv}, 2017.
\newblock Preprint, \url{https://arxiv.org/abs/1712.05142}.

\bibitem{NestedPolytopesER}
Michael~Gene Dobbins, Andreas Holmsen, and Tillmann Miltzow.
\newblock A universality theorem for nested polytopes.
\newblock {\em arXiv}, 1908.02213, 2019.

\bibitem{Kempe1876}
A.~B. Kempe.
\newblock On a general method of describing plane curves of the $n^{
  \text{th}}$ degree by linkwork.
\newblock 1-7(1):213--216, 1876.
\newblock Proceedings of the London Mathematical Society.

\bibitem{AnnaPreparation}
Anna Lubiw, Tillmann Miltzow, and Debajyoti Mondal.
\newblock The complexity of drawing a graph in a polygonal region.
\newblock In {\em International Symposium on Graph Drawing and Network
  Visualization}, pages 387--401. Springer, 2018.
\newblock arxiv: 1802.06699.

\bibitem{matousek2014intersection}
Ji{\v{r}}{\'{\i}} Matou{\v{s}}ek.
\newblock Intersection graphs of segments and $\exists \mathbb{R}$.
\newblock 2014.
\newblock Preprint, \url{https://arxiv.org/abs/1406.2636}.

\bibitem{mnev1988universality}
Nicolai~E Mn{\"e}v.
\newblock The universality theorems on the classification problem of
  configuration varieties and convex polytopes varieties.
\newblock In Oleg~Y. Viro, editor, {\em Topology and geometry -- Rohlin
  seminar}, pages 527--543. Springer-Verlag Berlin Heidelberg, 1988.

\bibitem{Schaefer2010}
Marcus Schaefer.
\newblock Complexity of some geometric and topological problems.
\newblock In {\em Proceedings of the 17th International Symposium on Graph
  Drawing (GD 2009)}, volume 5849 of {\em Lecture Notes in Computer Science
  (LNCS)}, pages 334--344. Springer, 2009.

\bibitem{schaefer2013realizability}
Marcus Schaefer.
\newblock Realizability of graphs and linkages.
\newblock In J\'{a}nos Pach, editor, {\em Thirty Essays on Geometric Graph
  Theory}, chapter~23, pages 461--482. Springer-Verlag New York, 2013.

\bibitem{DBLP:journals/mst/SchaeferS17}
Marcus Schaefer and Daniel \v{S}tefankovi\v{c}.
\newblock Fixed points, {N}ash equilibria, and the existential theory of the
  reals.
\newblock {\em Theory of Computing Systems}, 60(2):172--193, 2017.

\bibitem{shor1991stretchability}
Peter~W. Shor.
\newblock Stretchability of pseudolines is {NP}-hard.
\newblock In Peter Gritzmann and Bernd Sturmfels, editors, {\em Applied
  Geometry and Discrete Mathematics: The Victor Klee Festschrift}, volume~4 of
  {\em DIMACS -- Series in Discrete Mathematics and Theoretical Computer
  Science}, pages 531--554. American Mathematical Society and Association for
  Computing Machinery, 1991.

\end{thebibliography}



\end{document}